%% file: gand17.tex
\documentclass[copyright,creativecommons,noderivs,noncommercial]{eptcs}
\providecommand{\event}{GandALF 2017} 
\usepackage{underscore}           

\usepackage{hhline}
\usepackage{multirow}
\usepackage{paralist}
\usepackage{tikz}
\usetikzlibrary{positioning,arrows}
\usepackage{mathtools}

\usepackage{amsmath,amsfonts,amssymb,amsthm}

\DeclareMathOperator{\PTIME}{\mathbf{P}}
\DeclareMathOperator{\NP}{\mathbf{NP}}
\DeclareMathOperator{\Psp}{\mathbf{PSPACE}}

\DeclareMathOperator{\EXPSPACE}{\mathbf{EXPSPACE}}

\DeclareMathOperator{\LINAEXPTIME}{\mathbf{AEXP_{pol}}}

\DeclareMathOperator{\co}{\mathbf{co-}\!}
\newcommand{\Th}{\PTIME^{\NP[O(\log n)]}}
\newcommand{\Thsq}{\PTIME^{\NP[O(\log^2 n)]}}

\DeclareMathOperator{\Pref}{Pref}
\DeclareMathOperator{\Suff}{Suff}

\DeclareMathOperator{\lst}{lst}
\DeclareMathOperator{\fst}{fst}

\DeclareMathOperator{\hsA}{\langle A\rangle}
\DeclareMathOperator{\hsL}{\langle L\rangle}
\DeclareMathOperator{\hsB}{\langle B\rangle}
\DeclareMathOperator{\hsE}{\langle E\rangle}
\DeclareMathOperator{\hsD}{\langle D\rangle}
\DeclareMathOperator{\hsO}{\langle O\rangle}
\DeclareMathOperator{\hsX}{\langle X\rangle}

\DeclareMathOperator{\hsAt}{\langle \overline{A}\rangle}
\DeclareMathOperator{\hsLt}{\langle \overline{L}\rangle}
\DeclareMathOperator{\hsBt}{\langle \overline{B}\rangle}
\DeclareMathOperator{\hsEt}{\langle \overline{E}\rangle}
\DeclareMathOperator{\hsDt}{\langle \overline{D}\rangle}
\DeclareMathOperator{\hsOt}{\langle \overline{O}\rangle}

\newcommand{\A}{\mathsf{A}}
\newcommand{\Abar}{\mathsf{\overline{A}}}
\newcommand{\AAbar}{\mathsf{A\overline{A}}}

\newcommand{\AAbarBBbar}{\mathsf{A\overline{A}B\overline{B}}}
\newcommand{\AAbarEEbar}{\mathsf{A\overline{A}E\overline{E}}}
\newcommand{\ABbar}{\mathsf{A\overline{B}}}
\newcommand{\Bbar}{\mathsf{\overline{B}}}
\newcommand{\Ebar}{\mathsf{\overline{E}}}
\newcommand{\B}{\mathsf{B}}
\newcommand{\E}{\mathsf{E}}

\newcommand{\AbarE}{\mathsf{\overline{A}E}}
\newcommand{\BEbar}{\mathsf{B\overline{E}}}

\newcommand{\AAbarBBbarEbar}{\mathsf{A\overline{A}B\overline{B}\overline{E}}}

\newcommand{\AAbarEBbarEbar}{\mathsf{A\overline{A}E\overline{B}\overline{E}}}

\newcommand{\HSprop}{\mathsf{Prop}}

\newcommand{\Nat}{{\mathbb{N}}}
\newcommand{\Trans}{\textit{R}}

\DeclareMathAlphabet{\mathpzc}{OT1}{pzc}{m}{it}

\newcommand{\epist}{KC}

\newcommand{\Instance}{\mathcal{I}}

\newcommand{\Init}{\textit{Init}}

\newcommand{\Ku}{\ensuremath{\mathpzc{K}}}
\newcommand{\Prop}{\mathpzc{AP}}
\newcommand{\Lab}{\mu}
\newcommand{\GLab}{\textit{Lab}}
\newcommand{\tpl}[1]{(#1)}
\newcommand{\Lang}{{\mathcal{L}}}
\newcommand{\M}{{\mathcal{M}}}
\newcommand{\WS}{{\mathcal{W}}}

\newcommand{\FMC}{\mathsf{FMC}}
\newcommand{\RE}{\mathsf{RE}}
\newcommand{\acc}{\textit{acc}}

\DeclareMathOperator{\depthb}{d_B}

\theoremstyle{plain}
    \newtheorem{proposition}{Proposition}
    \newtheorem{theorem}[proposition]{Theorem}
    \newtheorem{lemma}[proposition]{Lemma}

\theoremstyle{definition}
    \newtheorem{definition}[proposition]{Definition}
    
    \newtheorem{property}[proposition]{Property}

\theoremstyle{remark}
    \newtheorem{remark}[proposition]{Remark}

\newcommand{\NFA}{\text{\sffamily NFA}}

\newcommand{\HS}{\text{\sffamily HS}}

\newcommand{\CTLStar}{\text{\sffamily CTL$^{*}$}}

\newcommand{\Au}{\ensuremath{\mathcal{A}}}
\newcommand{\SPEC}{\textsf{spec}}
\newcommand{\SD}{\textsf{SD}}

\newcommand{\Summary}{\ensuremath{\mathcal{S}}}
\newcommand{\PrefS}{{\mathit{PS}}}

\newcommand{\End}{\textit{end}}

\newcommand{\IMT}{\textit{IMT}}

\newcommand{\Pos}{\textit{Pos}}
\newcommand{\PNF}{\textit{PNF}}
\newcommand \tuple[1]{\langle #1 \rangle}

\newcommand{\AltN}{\Upsilon}

\newdimen\boxfigwidth 

\def\bigbox{\begingroup
  \boxfigwidth=\hsize
  \advance\boxfigwidth by -2\fboxrule
  \advance\boxfigwidth by -2\fboxsep
  \setbox4=\vbox\bgroup\hsize\boxfigwidth
  \hrule height0pt width\boxfigwidth\smallskip%
  \linewidth=\boxfigwidth
}
\def\endbigbox{\smallskip\egroup\fbox{\box4}\endgroup}

\title{On the Complexity of Model Checking for\\
Syntactically Maximal Fragments of\\ 
the Interval Temporal Logic HS with Regular Expressions\thanks{
The work by Alberto Molinari and Angelo Montanari has been supported by the GNCS project \emph{Logic and Automata for Interval Model Checking}.
}}

\author{Laura Bozzelli \qquad Adriano Peron
\institute{University of Napoli ``Federico II'', Napoli, Italy}
\email{lr.bozzelli@gmail.com \qquad adrperon@unina.it}
\and
Alberto Molinari \qquad Angelo Montanari
\institute{University of Udine, Udine, Italy}
\email{molinari.alberto@gmail.com \qquad angelo.montanari@uniud.it}
}
\def\titlerunning{Complexity of Model Checking for Maximal Fragments of HS with Regular Expressions}
\def\authorrunning{L. Bozzelli, A. Molinari, A. Montanari \& A. Peron}

\begin{document}
\maketitle

\begin{abstract}
\input{abstract}
\end{abstract}

\input{intro}
\input{Preliminaries}

\input{TrackProperty}
\input{ModelCheckingAAbarBBbarEbar}
\input{concl}

\bibliographystyle{eptcs}
\bibliography{bib2}

\end{document}

%% file: abstract.tex
%
%
%

In this paper, we investigate the model checking (MC) problem for Halpern and Shoham's interval temporal logic $\HS$.
In the last years,  interval temporal logic MC has received an increasing attention as a viable alternative to the traditional (point-based) temporal logic MC, which can be recovered as a special case.
Most results have been obtained under the homogeneity assumption,
that constrains a proposition letter to hold over an interval if and only if it holds over each component state.
Recently, Lomuscio and Michaliszyn proposed a way to relax such an assumption by exploiting regular expressions to define the behaviour of proposition letters over intervals in terms of their component states. 
When homogeneity is assumed, the exact complexity of MC is a difficult open question for full $\HS$ and for its two syntactically maximal fragments $\A\Abar\B\Bbar\Ebar$ and $\A\Abar\E\Bbar\Ebar$. 
In this paper, we provide an asymptotically
optimal bound to the complexity of these two fragments under the more expressive semantic variant based on regular expressions by showing that their MC problem is $\LINAEXPTIME$-complete, where $\LINAEXPTIME$ denotes the complexity class of problems decided by exponential-time bounded alternating Turing Machines making a polynomially bounded number of alternations. 

%% file: intro.tex
\section{Introduction}\label{sect:intro}
 Model checking (MC), which allows one to automatically check whether a model of a given system satisfies a desired behavioural property, is commonly recognized as one of the most effective techniques in automatic system verification.
Besides in formal verification, it has been successfully used also in more general contexts 
(e.g., databases, planning, configuration systems, multi-agent systems \cite{DBLP:conf/ecp/GiunchigliaT99,DBLP:conf/tacas/LomuscioR06}).
%
 The actual possibility of exploiting MC relies on a good balance of expressiveness and complexity in the choice of the system model and of the language for specifying behavioural properties. Systems are usually modeled as finite state-transition graphs (finite Kripke structures), while properties are commonly expressed by  formulas of point-based temporal logics, such as LTL, CTL, and CTL$^{*}$~\cite{pnueli1977temporal,emerson1986sometimes}.


In this paper, we focus on MC with interval temporal logic (ITL) as the specification language. ITL features intervals, instead of points, as its primitive temporal entities~\cite{HS91,digitalcircuitsthesis,Ven90}. 
ITL allows one to deal with relevant temporal properties, such as actions with duration, accomplishments, and temporal aggregations, which are inherently ``interval-based'' and cannot be properly expressed by point-based temporal logics. ITL has been fruitfully applied in various areas of computer science, including formal verification, computational linguistics, planning, and multi-agent systems~\cite{digitalcircuitsthesis,DBLP:journals/ai/PrattHartmann05,LM13}.



%
Among ITLs, the landmark is \emph{Halpern and Shoham's modal logic of time intervals} $\HS$~\cite{HS91}, which features one modality for each of the 13 ordering relations between pairs of intervals (the so-called Allen's relations~\cite{All83}), apart from equality.
(Actually, the three Allen's modalities \emph{meets} $\A$, \emph{started-by} $\B$, and \emph{finished-by} $\E$, together with the corresponding inverse modalities $\Abar$, $\Bbar$, and $\Ebar$, suffice for expressing the entire set of relations.)
The satisfiability problem for $\HS$ is undecidable over all relevant classes of linear orders~\cite{HS91}, and most of its fragments (with meaningful exceptions) are undecidable as well~\cite{DBLP:journals/amai/BresolinMGMS14,MM14}. 

The MC problem for $\HS$ and its fragments consists in the verification of the correctness of the behaviour of a given system with respect to interval properties expressed in $\HS$.
Each finite computation path is interpreted as an interval, and its labelling is defined on the basis of the labelling of the states occurring in the path. 
Most results have been obtained by imposing suitable restrictions on proposition letters labeling intervals: either a proposition letter can be constrained to hold over an interval if and only if it holds over each component state (\emph{homogeneity assumption}~\cite{Roe80}), or interval labeling can be defined in terms of the labeling of interval endpoints.

An almost complete picture of the MC problem for full $\HS$ and its fragments has been  recently depicted with the contribution of many works by Molinari et al.\ \cite{MMMPP15,MMP15B,MMP15,bmmps16,bmmps16b,MMMPP15,MMPS16}, which all consider MC over finite Kripke structures for $\HS$ endowed with a state-based semantics (allowing branching both in the past and in the future) enforcing the homogeneity assumption.
The summary of these results is depicted in the second column of Table~\ref{fig:overv} (the first column reports the fragments of $\HS$ denoted by the list of the featured modalities). The complexity classes shown in red represent new (upper/lower) bounds to the complexity of the problem deriving from the results of this paper, while the other classes (in black) are known bounds.
Only few, hard issues are left open in this picture, mostly regarding the precise complexity of the full logic and its maximal fragments.
A comparison of different semantic solutions (i.e., state-based semantics, trace-based semantics and computation-tree-based semantics), together with an expressiveness comparison with standard point-based temporal logics LTL, CTL, and CTL$^{*}$ can be found in~\cite{bmmps16c}.

Different assumptions have been done by Lomuscio and Michaliszyn in~\cite{LM13,LM14} for some $\HS$ fragments extended with epistemic operators ($\epist$). They assume a computation-tree-based semantics (formulae are interpreted over the unwinding of the Kripke structure) and interval labeling takes into account only the endpoints of intervals.  The different semantic assumptions prevent any immediate comparison with respect to the former approach. The decidability status of MC for full epistemic $\HS$ is still unknown. (A summary of the results by Lomuscio and Michaliszyn is depicted in the last column of Table~\ref{fig:overv}.) 

The first meaningful attempt to relax the homogeneity assumption can be found in~\cite{lm15}, where Lomuscio and Michaliszyn propose to use regular expressions to define the labeling of proposition letters over intervals in terms of the component states. Note that the homogeneity assumption can be trivially encoded  by regular expressions. In that work, the authors prove the decidability of MC with regular expressions for some very restricted fragments of epistemic $\HS$, giving some rough upper bounds to its computational complexity.  
A deeper insight into the problem of MC for $\HS$ with regular expressions can be found in~\cite{sefm2017} where, under the assumption of a state-based semantics, it is proved that MC with regular expressions for full $\HS$ is decidable, and that a large class of $\HS$ fragments  
can be checked in polynomial working space 
(see the third column of Table~\ref{fig:overv}). 

In this paper, we study the problems of MC for the two (syntactically) maximal (symmetric) fragments $\A\Abar\B\Bbar\Ebar$ and $\A\Abar\E\Bbar\Ebar$ with regular expressions, which are not covered by \cite{sefm2017}, proving that the complexity of both problems is $\LINAEXPTIME$-complete. $\LINAEXPTIME$ denotes the complexity class of problems decided by   exponential-time bounded alternating Turing Machines with a polynomially bounded number of alternations. 
Such a class captures the precise
complexity of some relevant problems \cite{tcs15l,FR75} (e.g., the first-order theory of real addition with order \cite{FR75}).
First, we note that  settling the exact complexity of these fragments under the homogeneity assumption (which can be encoded by regular expressions) is a difficult open question \cite{MMP15}. Moreover,  considering that $\LINAEXPTIME \subseteq \EXPSPACE$ and that $\HS$ under homogeneity  is subsumed by $\HS$ with regular expressions, the results proved in this paper improve  the upper bounds for the fragments  $\A\Abar\B\Bbar\Ebar$ and $\A\Abar\E\Bbar\Ebar$ given in 
\cite{MMP15}.

These results are obtained by preliminarily establishing an \emph{exponential-size model-trace property}: for each interval,  it is possible to find an interval of bounded exponential length that is indistinguishable with respect to the fulfillment of $\A\Abar\B\Bbar\Ebar$ formulas (resp., $\A\Abar\E\Bbar\Ebar$).
Such a property allows us to devise a MC procedure belonging to the class $\LINAEXPTIME$. 
Finally, the 
matching lower bounds are obtained 
by polynomial-time reductions from the so-called \emph{alternating multi-tiling problem}, and they already hold for the fragments $\B\Ebar$ and $\E\Bbar$ of $\A\Abar\B\Bbar\Ebar$ and $\A\Abar\E\Bbar\Ebar$, respectively. 

\begin{table}[t]
\centering
\caption{Complexity of MC for $\HS$ and its fragments ($^\dagger$local MC).}\label{fig:overv}
\input{overvTable}
\end{table}
%

The paper is structured as follows. In Section~\ref{sec:backgr}, we introduce the logic $\HS$ and provide some background knowledge. 
In Section~\ref{sec:AAbarBBbarEbarTrackProperty} we prove the exponential-size model-trace property for $\A\Abar\B\Bbar\Ebar$. In Section~\ref{sec:UpperBound}, we provide an $\LINAEXPTIME$ upper bound to the MC problem for $\AAbarBBbarEbar$.  Finally, in Section~\ref{sec:LowerBound}, we prove the hardness of the fragment $\B\Ebar$. Similar proofs can be given for establishing the $\LINAEXPTIME$-completeness of $\A\Abar\E\Bbar\Ebar$, and the $\LINAEXPTIME$-hardness of $\E\Bbar$.

Due to space constraints, most of the proofs are omitted here: they can be found in \cite{preprintGand}.

%% file: overvTable.tex
\begingroup
\renewcommand*{\arraystretch}{1.3}
\begin{tabular}{|@{\ }c@{\ }|@{\ }c@{\ }|@{\ }c@{\ }||@{\ }c@{\ }|}
\hline 
 & Homogeneity & Regular expressions & \cite{LM13} -- \cite{lm15}\\
\hline \hline 
\multirow{2}{*}{Full $\HS$, $\B\E$} & non-elem.  & non-elem.  & $\B\E$+$\epist^\dagger$: $\Psp$ \\
 & $\EXPSPACE$-hard & $\EXPSPACE$-hard & $\B\E^\dagger$: $\PTIME$\\
\hline 
\multirow{2}{*}{$\A\Abar\B\Bbar\Ebar,\A\Abar\E\Bbar\Ebar$} & $\in\EXPSPACE$ [\textcolor{red}{$\in\LINAEXPTIME$}]& non-elem  $\Psp$-hard &\\
 & $\Psp$-hard & [\textcolor{red}{$\LINAEXPTIME$-complete}]& \\
\hline 
\multirow{2}{*}{$\AAbar\Bbar\Ebar$} & \multirow{2}{*}{$\Psp$-complete} & non-elem [\textcolor{red}{$\in\LINAEXPTIME$}] & \\
 & & $\Psp$-hard & \\
\hline
$\AAbarBBbar,\B\Bbar,\Bbar,$ & \multirow{2}{*}{$\Psp$-complete} & \multirow{2}{*}{$\Psp$-complete} & \multirow{2}{*}{$\A\Bbar$+$\epist$: non-elem.}\\
$\AAbarEEbar,\E\Ebar,\Ebar$ & & & \\
\hline 
$\AAbar\B,\AAbar\E,\A\B,\Abar\E$ & $\PTIME^{\NP}\!$-complete & $\Psp$-complete & \\
\hline 
\multirow{2}{*}{$\A\Abar,\Abar\B,\A\E,\A,\Abar$} & $\in\Thsq$ & \multirow{2}{*}{$\Psp$-complete} & \\
 & $\Th$-hard &  & \\
\hline 
$\HSprop, \B,\E$ & $\co\NP$-complete & $\Psp$-complete & \\
\hline 
\end{tabular}

\endgroup

%% file: Preliminaries.tex
\section{Preliminaries}\label{sec:backgr}

We introduce some preliminary notation.
Let $\Nat$ be the set of natural numbers. For all  $i,j\in\Nat $, with $i\leq j$, $[i,j]$ denotes the set of natural numbers $h$ such that $i\leq h\leq j$.

Let $\Sigma$ be an alphabet and $w$ be a  finite  word over $\Sigma$. We denote by $|w|$ the length of $w$. By $\varepsilon$ we denote the empty word. For all  $1\leq i\leq j\leq |w|$, $w(i)$ denotes the
$i$-th letter of $w$, while $w(i,j)$ denotes the finite subword of $w$ given by $w(i)w(i+1)\cdots w(j)$. For $|w|=n$, we define
$\fst(w)=w(1)$ and $\lst(w)=w(n)$. The sets of all proper prefixes and suffixes of $w$
are $\Pref(w)=\{w(1,i) \mid 1\leq i\leq n-1\}$ and $\Suff(w)=\{w(i,n)\mid 2\leq i\leq n\}$, respectively.
The concatenation of two words $w$ and $w'$ is denoted as usual by $w\cdot w'$. Moreover, if $\lst(w)=\fst(w')$, $w\star w'$ represents $w(1,n-1)\cdot w'$, where $n=|w|$ ($\star$-concatenation). 

\subsection{Kripke structures, regular expressions, and finite automata}

Finite state systems are usually modelled as finite Kripke structures. Let $\Prop$ be a finite set of proposition letters, which represent predicates decorating the states of the given system.

\begin{definition}[Kripke structure]
A \emph{Kripke structure} over  $\Prop$  is a tuple $\Ku=(\Prop,S, \Trans,\mu,s_0)$, where  $S$ is a set of states,
$\Trans\subseteq S\times S$ is a
transition relation,
$\mu:S\mapsto 2^\Prop$ is a total labelling function assigning to each state $s$ the set of propositions that hold over it, and $s_0\in S$ is the initial state. $\Ku$ is said finite if $S$ is finite.
\end{definition}



Let $\Ku=(\Prop,S, \Trans,\mu,s_0)$ be a Kripke structure.
A \emph{trace} (or finite path) of $\Ku$ is a non-empty finite word $\rho$ over $S$ such that $(\rho(i),\rho(i+1))\in \Trans$ for all $i\in [1,|\rho|-1]$. A trace is \emph{initial} if it starts from the initial state $s_0$.
A trace $\rho$ induces the finite word  $\mu(\rho)$ over $2^{\Prop}$ given by
 $\mu(\rho(1))\cdots \mu(\rho(n))$ with $n=|\rho|$. We call $\mu(\rho)$ the \emph{labeling sequence induced by} $\rho$.

Let us recall now the class of regular expressions over finite words. Since we are interested in expressing requirements over the labeling sequences
induced by the traces of Kripke structures, which are finite words over $2^{\Prop}$, here we consider \emph{propositional-based} regular expressions ($\RE$), where the atomic expressions are propositional formulas over $\Prop$ instead of letters over an alphabet. Formally, the set of $\RE$ $r$ over $\Prop$ is defined as \[ r ::= \varepsilon\;\vert\; \phi\;\vert\; r\cup r\;\vert\; r\cdot r\;\vert\; r^{*} ,\]
%
%
where $\phi$ is a propositional formula over $\Prop$. The size $|r|$ of an $\RE$ $r$ is the number of subexpressions of $r$.
 An $\RE$ $r$ denotes a language $\Lang(r)$ of finite words over $2^{\Prop}$ defined as:
\begin{compactitem}
  \item $\Lang(\varepsilon)=\{\varepsilon\}$;
  \item $\Lang(\phi)=\{A\in 2^{\Prop}\mid A \text{ satisfies }\phi\}$;
  \item $\Lang(r_1\cup r_2)=\Lang(r_1)\cup \Lang(r_2)$;
  \item $\Lang(r_1\cdot r_2)=\Lang(r_1)\cdot \Lang(r_2)$;
  \item $\Lang(r^{*})=(\Lang(r))^{*}$.
\end{compactitem}

We also recall the class of nondeterministic finite automata  over finite words (\NFA).  An $\NFA$ is a tuple
$\Au =\tpl{\Sigma,Q,Q_0,\Delta,F}$, where $\Sigma$ is a finite alphabet, $Q$ is a finite set of states, $Q_0\subseteq Q$ is the set of initial states,
$\Delta\subseteq Q\times \Sigma \times Q$ is the transition relation, and $F\subseteq Q$ is the set of accepting states.
An $\NFA$ $\Au$ is \emph{complete} if, for all $(q,\sigma)\in Q\times \Sigma$, $(q,\sigma,q')\in \Delta$ for some $q'\in Q$. Given a finite word $w$ over $\Sigma$ with $|w|=n$  and two states $q,q'\in Q$, a run of $\Au$ from $q$ to $q'$ over $w$ is a sequence of states $q_1,\ldots,q_{n+1}$ such that $q_1=q$, $q_{n+1}=q'$, and for all $i\in [1,n]$, $(q_i,w(i),q_{i+1})\in\Delta$. The language $\Lang(\Au)$  accepted by $\Au$ is the set of finite words $w$ on $\Sigma$ such that there is a
run from some initial state to some accepting state over $w$.

\begin{remark} Given a $\RE$ $r$, by a standard construction \cite{kleene56representation}, one can compositionally construct a complete $\NFA$ $\Au_r$ with alphabet $2^{\Prop}$, whose number of states is linear in the size of $r$. We call $\Au_r$ the \emph{canonical} $\NFA$ associated with $r$.
\end{remark}

\subsection{The interval temporal logic $\HS$}
A systematic logical study of interval representation and reasoning was proposed by J.\ Y.\ Halpern and Y.\ Shoham, who introduced the interval temporal logic $\HS$~\cite{HS91} featuring one modality for each Allen
relation~\cite{All83}, but equality.
Table~\ref{allen} depicts 6 of the 13 Allen's relations,
together with the corresponding $\HS$ (existential) modalities. The other 7 relations are the 6 inverse relations (given a binary relation $\mathpzc{R}$, its inverse $\overline{\mathpzc{R}}$ is such that $b \overline{\mathpzc{R}} a$ iff $a \mathpzc{R} b$) and equality.

\begin{table}[tb]
\renewcommand{\arraystretch}{1.2}
\centering
\caption{Allen's relations and corresponding $\HS$ modalities.}\label{allen}
\begin{tabular}{cclc}
\hline
\rule[-1ex]{0pt}{3.5ex} Allen relation & $\HS$ & Definition w.r.t. interval structures &  Example\\
\hline

&   &   & \multirow{7}{*}{\input{allensRels.tex}}\\

\textsc{meets} & $\hsA$ & $[x,y]\mathpzc{R}_A[v,z]\iff y=v$ &\\

\textsc{before} & $\hsL$ & $[x,y]\mathpzc{R}_L[v,z]\iff y<v$ &\\

\textsc{started-by} & $\hsB$ & $[x,y]\mathpzc{R}_B[v,z]\iff x=v\wedge z<y$ &\\

\textsc{finished-by} & $\hsE$ & $[x,y]\mathpzc{R}_E[v,z]\iff y=z\wedge x<v$ &\\

\textsc{contains} & $\hsD$ & $[x,y]\mathpzc{R}_D[v,z]\iff x<v\wedge z<y$ &\\

\textsc{overlaps} & $\hsO$ & $[x,y]\mathpzc{R}_O[v,z]\iff x<v<y<z$ &\\

\hline
\end{tabular}
\end{table}

Given a finite set $\mathpzc{P}_u$  of \emph{uninterpreted interval properties},
the $\HS$ language over $\mathpzc{P}_u$ consists of propositions from $\mathpzc{P}_u$, the Boolean connectives $\neg$ and $\wedge$,
and a temporal modality for each of the (non trivial) Allen's relations, i.e., $\hsA$, $\hsL$, $\hsB$, $\hsE$, $\hsD$, $\hsO$, $\hsAt$, $\hsLt$, $\hsBt$, $\hsEt$, $\hsDt$, and $\hsOt$.
$\HS$ formulas are defined by the grammar \[\psi ::= p_u \;\vert\; \neg\psi \;\vert\; \psi \wedge \psi \;\vert\; \langle X\rangle\psi,\]
where $p_u\in\mathpzc{P}_u$ and $X\in\{A,L,B,E,D,O,\overline{A},\overline{L},\overline{B},\overline{E},\allowbreak\overline{D},\overline{O}\}$.
We  also exploit the standard logical connectives (disjunction $\vee$ and implication $\rightarrow$) 
 as abbreviations.
Furthermore, for any existential modality $\langle X\rangle$, the dual universal modality $[X]\psi$ is defined as $\neg\langle X\rangle\neg\psi$. 

An $\HS$ formula $\varphi$ is in  \emph{positive normal form} (\PNF) if negation is applied only to atomic formulas in $\mathpzc{P}_u$.
By using De Morgan's laws and for any existential modality $\tuple{X}$, the dual universal modality $[X]$,  we can  convert in linear-time an $\HS$ formula $\varphi$ into an equivalent formula in \PNF, called the \PNF\ of $\varphi$. For a formula   $\varphi$ in \PNF, the \emph{dual}
$\widetilde{\varphi}$ of  $\varphi$ is the \PNF\ of $\neg\varphi$.

Given any subset of Allen's relations $\{X_1,\ldots,X_n\}$, we denote by $\mathsf{X_1 \cdots X_n}$ the $\HS$ fragment closed under Boolean connectives that features (existential and universal) 
modalities for $X_1 ,\ldots, X_n$ only.

Without loss of generality, we assume the \emph{non-strict semantics of $\HS$}, which admits intervals consisting of a single point. (All the results we prove in the paper hold for the strict semantics as well.) Under such an assumption, all $\HS$ modalities can be expressed in terms of modalities
$\hsB, \hsE, \hsBt$, and $\hsEt$ \cite{Ven90}. 
$\HS$ can, thus, be viewed as a multi-modal logic with
4 primitive modalities. However, since  we  focus on the $\HS$ fragments $\AAbarEBbarEbar$ and $\AAbarBBbarEbar$, that do not feature $\hsB$  and $\hsE$ respectively, we also consider the modalities $\hsA$ and $\hsAt$. Note that the modalities $\hsL$ and $\hsO$ (resp., $\hsLt$ and $\hsOt$) can be expressed in the fragment $\AAbarEBbarEbar$ (resp., $\AAbarBBbarEbar$).

As for the semantics of $\HS$, in this paper we follow the approach of~\cite{sefm2017}, 
where the intervals correspond  to the traces of a finite Kripke  structure $\Ku$ (\emph{state-based semantics}) and
each abstract interval proposition $p_u\in \mathpzc{P}_u$ denotes a regular language of finite words over $2^{\Prop}$. More specifically, every
abstract interval proposition  $p_u$ is a (propositional-based) regular expression over $\Prop$. 
Thus, in the following, for the sake of simplicity, by an  $\HS$ formula over $\Prop$ we mean an $\HS$ formula whose abstract interval propositions (or atomic formulas) are $\RE$ over $\Prop$.

Given a Kripke structure
$\Ku=(\Prop,S, E,\mu,s_0)$ over $\Prop$, a trace $\rho$ of $\Ku$, and an $\HS$ formula $\varphi$ over $\Prop$, the satisfaction relation $\Ku,\rho\models \varphi$ is inductively defined as follows (we omit the standard clauses for the Boolean connectives):
\[ \begin{array}{ll}
\Ku,\rho   \models r  &  \Leftrightarrow  \mu(\rho) \in \Lang(r) \text{ for each $\RE$ $r$ over $\Prop$},\\
\Ku,\rho  \models \hsB\varphi  & \Leftrightarrow   \text{there exists } \rho'\in\Pref(\rho) \text{ such that }\Ku,\rho'   \models \varphi,\\
\Ku,\rho  \models \hsE\varphi  & \Leftrightarrow   \text{there exists } \rho'\in\Suff(\rho) \text{ such that }\Ku,\rho'   \models \varphi,\\
\Ku,\rho   \models \hsBt\varphi  & \Leftrightarrow   \Ku,\rho'   \models \varphi \text{ for some trace } \rho' \text{ such that } \rho\in\Pref(\rho') ,\\
\Ku,\rho   \models \hsEt\varphi  & \Leftrightarrow   \Ku,\rho'   \models \varphi \text{ for some trace } \rho' \text{ such that } \rho\in\Suff(\rho') ,\\
\Ku,\rho   \models \hsA\varphi  & \Leftrightarrow   \Ku,\rho'   \models \varphi \text{ for some trace } \rho' \text{ such that } \fst(\rho')=\lst(\rho) ,\\
\Ku,\rho   \models \hsAt\varphi  & \Leftrightarrow   \Ku,\rho'   \models \varphi \text{ for some trace } \rho' \text{ such that } \lst(\rho')=\fst(\rho).
\end{array} \]

$\Ku$ is a \emph{model} of $\varphi$, denoted $\Ku\models \varphi$, if for all initial traces $\rho$ of $\Ku$, it holds that $\Ku,\rho\models \varphi$. The MC problem for $\HS$ is checking, for a finite Kripke structure $\Ku$ and an $\HS$ formula $\varphi$, whether $\Ku\models \varphi$ or not. 

Note that the state-based semantics provides a branching-time setting both in the past and in the future. In particular, while the modalities for $B$ and $E$ are linear-time (they allow us to select prefixes and suffixes of the current trace), the modalities for $A$ and $\overline{B}$ (resp., $\overline{A}$ and $\overline{E}$) are branching-time in the future (resp., in the past) since they allow us to nondeterministically extend a trace in the future (resp., in the past). As shown in \cite{bmmps16c}, for the considered semantics, the logics $\HS$ and $\CTLStar$ are expressively incomparable already under the homogeneity assumption. However, under the homogeneity assumption, the use of  the past branching-time modalities $\overline{A}$ and $\overline{E}$ is necessary for capturing requirements which cannot be expressed in $\CTLStar$. For instance, the requirement 
``\emph{each state reachable from the initial one where $p$ holds has a predecessor where $p$ holds as well}'' cannot be expressed in \CTLStar, but can be easily expressed in the fragment $\AbarE$ \cite{bmmps16c}.  In the more expressive  setting based on regular expressions, the future branching-time modalities $A$ and $\overline{B}$ are already sufficient for capturing requirements which cannot be expressed in \CTLStar, such as the following branching-time bounded response property: ``\emph{for each state reachable from the initial one where a request  $\textit{req}$ occurs, there is a computation from this state such that the request is followed by a response $\textit{res}$ within an  \emph{even number} of steps}''. This requirement can be expressed in the  fragment $\ABbar$ as follows: $[A](\textit{req} \rightarrow \hsBt (\textit{req}\cdot (\top\cdot \top)^{*}\cdot \textit{res}))$.

In the rest of the paper, we focus on the fragment $\AAbarBBbarEbar$. Analogous constructions and results can be symmetrically given 
for the fragment $\AAbarEBbarEbar$ as well.

%% file: allensRels.tex
\begin{tikzpicture}[scale=1.15]
\draw[draw=none,use as bounding box](-0.3,0.2) rectangle (3.3,-3.1);
\coordinate [label=left:\textcolor{red}{$x$}] (A0) at (0,0);
\coordinate [label=right:\textcolor{red}{$y$}] (B0) at (1.5,0);
\draw[red] (A0) -- (B0);
\fill [red] (A0) circle (2pt);
\fill [red] (B0) circle (2pt);

\coordinate [label=left:$v$] (A) at (1.5,-0.5);
\coordinate [label=right:$z$] (B) at (2.5,-0.5);
\draw[black] (A) -- (B);
\fill [black] (A) circle (2pt);
\fill [black] (B) circle (2pt);

\coordinate [label=left:$v$] (A) at (2,-1);
\coordinate [label=right:$z$] (B) at (3,-1);
\draw[black] (A) -- (B);
\fill [black] (A) circle (2pt);
\fill [black] (B) circle (2pt);

\coordinate [label=left:$v$] (A) at (0,-1.5);
\coordinate [label=right:$z$] (B) at (1,-1.5);
\draw[black] (A) -- (B);
\fill [black] (A) circle (2pt);
\fill [black] (B) circle (2pt);

\coordinate [label=left:$v$] (A) at (0.5,-2);
\coordinate [label=right:$z$] (B) at (1.5,-2);
\draw[black] (A) -- (B);
\fill [black] (A) circle (2pt);
\fill [black] (B) circle (2pt);

\coordinate [label=left:$v$] (A) at (0.5,-2.5);
\coordinate [label=right:$z$] (B) at (1,-2.5);
\draw[black] (A) -- (B);
\fill [black] (A) circle (2pt);
\fill [black] (B) circle (2pt);

\coordinate [label=left:$v$] (A) at (1.3,-3);
\coordinate [label=right:$z$] (B) at (2.3,-3);
\draw[black] (A) -- (B);
\fill [black] (A) circle (2pt);
\fill [black] (B) circle (2pt);

\coordinate (A1) at (0,-3);
\coordinate (B1) at (1.5,-3);
\draw[dotted] (A0) -- (A1);
\draw[dotted] (B0) -- (B1);
\end{tikzpicture}

%% file: TrackProperty.tex
\section{Exponential-size model-trace property for $\AAbarBBbarEbar$}\label{sec:AAbarBBbarEbarTrackProperty}

In this section, we show an \emph{exponential-size model-trace property} for $\AAbarBBbarEbar$, which will be used as the basic step to prove that the MC problem for $\AAbarBBbarEbar$  belongs to $\LINAEXPTIME$.
Fix a Kripke structure $\Ku=(\Prop,S, \Trans,\mu,s_0)$ and a finite set $\SPEC=\{r_1,\ldots,r_H\}$ of (propositional-based)
regular expressions over $\Prop$:
such a property ensures that for each $h\geq 0$ and trace $\rho$ of $\Ku$,
 it is possible to build another trace $\rho'$ of $\Ku$, of bounded exponential length, which is indistinguishable from $\rho$ with respect to the fulfilment of any $\AAbarBBbarEbar$ formula $\varphi$ having atomic formulas in $\SPEC$ and nesting depth of the modality $\hsB$ at most $h$ (written $\depthb(\varphi)\leq h$).
 Formally, $\depthb(\varphi)$ is inductively defined as follows
  $(i)$~$\depthb(r)=0$, for any $\RE$ $r$ over $\mathpzc{AP}$;
        $(ii)$~$\depthb(\neg\psi)=\depthb(\psi)$;
        $(iii)$~$\depthb(\psi\wedge\phi)=\max\{\depthb(\psi),\depthb(\phi)\}$;
        $(iv)$~$\depthb(\hsB\psi)=1+\depthb(\psi)$;
        $(v)$~$\depthb(\hsX\psi)=\depthb(\psi)$, for $X\in\{A, \overline{A}, \overline{B}, \overline{E}\}$.

In order to state the result, we first introduce the notion of \emph{$h$-prefix bisimilarity}   between a pair of traces $\rho$ and $\rho'$ of $\Ku$. As proved by Proposition \ref{prop:fulfillmentPreservingPrefix} below, $h$-prefix bisimilarity   is a sufficient condition for two traces  $\rho$ and $\rho'$ to be indistinguishable with respect to the fulfillment of any $\AAbarBBbarEbar$ formula $\varphi$ over $\SPEC$ with $\depthb(\varphi)\leq h$.
Then, for a given trace $\rho$, we  show how to determine a subset of positions of $\rho$, called the \emph{$h$-prefix sampling} of $\rho$, that allows us to build another trace $\rho'$ having singly exponential length (both in $h$ and $|\SPEC|$, where $|\SPEC|$ is defined  as  $\sum_{r\in\SPEC}|r|$) such that $\rho$ and $\rho'$ are $h$-prefix bisimilar.

For any regular expression $r_\ell$ in $\SPEC$ with $\ell\in [1,H]$, let $\Au_\ell=\tpl{2^{\Prop},Q_\ell,Q_\ell^{0},\Delta_\ell,F_\ell}$ be the \emph{canonical (complete)} \NFA\ accepting $\Lang(r_\ell)$ (recall that
$|Q_\ell|\leq 2|r_\ell|$).
Without loss of generality, we assume that the sets of states of these automata are pairwise disjoint.  

The notion of prefix bisimilarity exploits the notion of \emph{summary}  of a trace $\rho$ of $\Ku$, namely a tuple ``recording'' the initial and final states of $\rho$, and, for each automaton $\Au_\ell$ with $\ell\in [1,H]$, the pairs of states $q,q'\in Q_\ell$  such that some run of $\Au_\ell$ over $\mu(\rho)$ goes from $q$ to $q'$.

\begin{definition}[Summary of a trace]
Let $\rho$ be a trace of $\Ku$ with $|\rho|=n$. The summary $\Summary(\rho)$ of $\rho$ (w.r.t.\ $\SPEC$) is the triple
$(\rho(1),\Pi,\rho(n))$, where $\Pi$ is the set of pairs $(q,q')$ such that there is $\ell\in [1,H]$ so that $q,q'\in Q_\ell$ and there is a run of $\Au_\ell$ from $q$ to $q'$ over $\Lab(\rho)$.
\end{definition}

\noindent Note that the number of summaries is at most $|S|^2\cdot 2^{(2|\SPEC|)^2}$. Evidently, the following holds. 

\begin{proposition}\label{prop:Summaries} Let $h\geq 0$, and $\rho$ and $\rho'$ be two traces of  $\Ku$ such that $\Summary(\rho)=\Summary(\rho')$. Then, for all regular expressions $r\in \SPEC$ and traces $\rho_L$ and $\rho_R$ of $\Ku$ such that $\rho_L \star \rho$ and  $\rho \star \rho_R$ are defined, the following hold:
(1) $\mu(\rho)\in\Lang(r)$ iff $\mu(\rho')\in\Lang(r)$;
(2) $\Summary(\rho_L\star \rho)=\Summary(\rho_L\star \rho')$;
(3) $\Summary(\rho\star \rho_R)=\Summary(\rho'\star \rho_R)$.
\end{proposition}

We now introduce the notion of \emph{prefix bisimilarity}   between a pair of traces $\rho$ and $\rho'$ of $\Ku$.

\begin{definition}[Prefix bisimilarity]
Let $h\geq 0$. Two traces $\rho$ and $\rho'$  of $\Ku$  are \emph{$h$-prefix bisimilar} (w.r.t.\ $\SPEC$) if the following conditions inductively hold:
\begin{compactitem}
  \item for $h=0$: $\Summary(\rho)=\Summary(\rho')$;
  \item for $h>0$: $\Summary(\rho)=\Summary(\rho')$  and for each proper prefix $\nu$ of $\rho$ (resp., proper prefix $\nu'$ of $\rho'$), there is
  a proper prefix $\nu'$ of $\rho'$ (resp., proper prefix $\nu$ of $\rho$) such that $\nu$ and $\nu'$ are $(h-1)$-prefix bisimilar.
\end{compactitem}
\end{definition}

\begin{property}
For all $h\geq 0$, $h$-prefix   bisimilarity is an equivalence relation over traces of $\Ku$.
\end{property}

The $h$-prefix  bisimilarity of two traces $\rho$ and $\rho'$ is preserved by right (resp., left) $\star$-concatenation with another trace of $\Ku$. 

\newcounter{prop-invarianceLeftRightPrefix}
\setcounter{prop-invarianceLeftRightPrefix}{\value{proposition}}

\begin{proposition}\label{prop:invarianceLeftRightPrefix} Let $h\geq 0$, and $\rho$ and $\rho'$ be two $h$-prefix  bisimilar traces of  $\Ku$. Then, for all traces $\rho_L$ and $\rho_R$ of $\Ku$ such that $\rho_L \star \rho$ and  $\rho \star \rho_R$ are defined, the following hold:
\newline
(1) $\rho_L\star \rho$ and $\rho_L\star \rho'$ are $h$-prefix   bisimilar; (2) $\rho\star \rho_R$ and $\rho'\star \rho_R$ are $h$-prefix  bisimilar.
\end{proposition}

By exploiting Propositions~\ref{prop:Summaries} and~\ref{prop:invarianceLeftRightPrefix},
we can prove that $h$-prefix  bisimilarity preserves the fulfillment of $\AAbarBBbarEbar$ formulas over $\SPEC$ having nesting depth of modality $\hsB$  at most $h$.

\begin{proposition}\label{prop:fulfillmentPreservingPrefix} Let $h\geq 0$, and $\rho$ and $\rho'$ be two $h$-prefix bisimilar traces of $\Ku$. Then, for each $\AAbarBBbarEbar$
formula $\psi$ over $\SPEC$ with $\depthb(\psi)\leq h$, we have
$\Ku,\rho\models\psi$ iff $\Ku,\rho'\models\psi$.
\end{proposition}
\begin{proof}
We prove the proposition by a nested induction on the structure of the formula $\psi$ and on the nesting depth $\depthb(\psi)$. For the base case, $\psi$ is a regular expression in $\SPEC$. Since $\Summary(\rho)=\Summary(\rho')$ ($\rho$ and $\rho'$ are $h$-prefix bisimilar) the result follows  by Proposition~\ref{prop:Summaries}.
Now, let us consider the inductive case.
The cases where the root modality of $\psi$ is a Boolean connective directly follow by the inductive hypothesis.
As for the cases where the root modality is either $\hsA$ or $\hsAt$, the result follows from the fact that, being $\rho$ and $\rho'$ $h$-prefix bisimilar,
$\fst(\rho)=\fst(\rho')$ and $\lst(\rho)=\lst(\rho')$.
It remains to consider the cases where the root modality is in $\{\hsB,\hsBt,\hsEt\}$. We prove the implication
$\Ku,\rho\models \psi \Rightarrow \Ku,\rho'\models \psi$ (the converse implication being similar). 
Let
$\Ku,\rho\models \psi$. 
\begin{compactitem}
  \item $\psi=\hsB\varphi$: since $0<\depthb(\psi)\leq h$,  it holds that $h>0$.  Since $\Ku,\rho\models\hsB\varphi$,   there is a proper prefix $\nu$ of $\rho$  such that $\Ku,\nu\models\varphi$.
  Since $\rho$ and $\rho'$ are $h$-prefix bisimilar, there is a proper prefix $\nu'$ of $\rho'$ such that $\nu$ and $\nu'$ are $(h-1)$-prefix bisimilar.
  Being $\depthb(\varphi)\leq h-1$, by the inductive hypothesis we obtain that  $\Ku,\nu'\models\varphi$. Hence, $\Ku,\rho'\models\hsB\varphi$: the thesis follows.
  \item $\psi=\hsBt\varphi$: since $\Ku,\rho\models\hsBt\varphi$, there is a trace $\rho_R$ such that $|\rho_R|>1$ and
$\Ku,\rho\star\rho_R\models\varphi$. By Proposition~\ref{prop:invarianceLeftRightPrefix}, $\rho\star\rho_R$ and $\rho'\star\rho_R$ are $h$-prefix bisimilar. By the inductive hypothesis on  the structure of the formula, we obtain that $\Ku,\rho'\star\rho_R\models\varphi$,  hence,
  $\Ku,\rho'\models\hsBt\varphi$. 
  \item    $\psi=\hsEt\varphi$: this case is similar to the previous one.\qedhere 
\end{compactitem}
\end{proof}

In the following, we show how a trace $\rho$, whose length exceeds a suitable exponential bound---precisely, $(|S|\cdot 2^{(2|\SPEC|)^2})^{h+2}$---can be contracted preserving $h$-prefix bisimilarity and, consequently, the fulfillment of formulas $\varphi$ with $\depthb(\varphi) \leq h$. The basic contraction step of $\rho$ is performed by choosing a subset of $\rho$-positions called $h$-\emph{prefix sampling} ($\PrefS_h$). A contraction can be performed whenever there are two positions $\ell < \ell'$ satisfying $\Summary(\rho(1,\ell))=\Summary(\rho(1,\ell'))$ in between two consecutive positions in the linear ordering of $\PrefS_h$. We  prove that  by taking the contraction $\rho'=\rho(1,\ell)\cdot \rho(\ell'+1,|\rho|)$, we obtain a trace of $\Ku$ which is $h$-prefix bisimilar to $\rho$. The basic contraction step can then be iterated over $\rho'$ until the length bound is reached.

The notion of $h$-prefix sampling is inductively defined using the notion of \emph{prefix-skeleton sampling}.
%
%
For a set $I$ of natural numbers, by ``two consecutive elements of $I$'' we refer to a pair of elements $i,j\in I$ such that $i<j$ and $I\cap [i,j]=\{i,j\}$.

\begin{definition}[Prefix-skeleton sampling]\label{def:skeleton}  Let $\rho$ be a trace of $\Ku$. Given two $\rho$-positions $i$ and $j$, with $i\leq j$, the \emph{prefix-skeleton sampling of $\rho$ in the interval $
[i,j]$} is the \emph{minimal} set $\Pos \supseteq \{i,j\}$ of $\rho$-positions in the interval $[i,j]$ satisfying the condition:
\begin{compactitem}
\item  for each $k\in [i+1,j-1]$,  the minimal position $k'\in [i+1,j-1]$ such that $\Summary(\rho(1,k'))=\Summary(\rho(1,k))$ is in $\Pos$.
\end{compactitem}
\end{definition}

It immediately follows from Definition \ref{def:skeleton} that the prefix-skeleton sampling $\Pos$  of (any) trace $\rho$ in an interval  $[i,j]$ of $\rho$-positions is such that $|\Pos|\leq (|S|\cdot 2^{(2|\SPEC|)^2})+2$. 

\begin{definition}[$h$-prefix sampling] Let $h\geq 0$.
The \emph{$h$-prefix sampling of a trace $\rho$} of $\Ku$ is the \emph{minimal} set $\PrefS_h$ of $\rho$-positions inductively satisfying the following conditions:
\begin{compactitem}
  \item Base case: $h=0$. $\PrefS_0=\{1,|\rho|\}$;
  \item Inductive step: $h> 0$.
  $(i)$  $\PrefS_h\supseteq\PrefS_{h-1}$ and
  $(ii)$ for all pairs of consecutive positions $i,j$ in $\PrefS_{h-1}$,
  the prefix-skeleton sampling of $\rho$ in the interval $[i,j]$ is in $\PrefS_h$.
\end{compactitem}

Let $i_1<\ldots<i_N$ be the ordered sequence of positions in $\PrefS_h$ (note that $i_1=1$ and $i_N=|\rho|$). The \emph{$h$-sampling word of $\rho$} is the sequence of summaries 
$\Summary(\rho(1,i_1))\cdots \Summary(\rho(1,i_N))$.
\end{definition}

The following upper bound to the cardinality of prefix samplings holds.

\begin{property}\label{property:prefSamBound}
The $h$-prefix sampling $\PrefS_h$ of a trace $\rho$ of $\Ku$ is such that $|\PrefS_h|\leq (|S|\cdot 2^{(2|\SPEC|)^2})^{h+1}$.
\end{property}

The following lemma 
states that, for two traces, the property of having the same $h$-sampling word is a sufficient condition to be $h$-prefix bisimilar.

\begin{lemma}\label{lemma:prefixSamplingOne} For $h\geq 0$, two traces
having the same
$h$-sampling word are $h$-prefix bisimilar.
\end{lemma}

By exploiting the sufficient condition of Lemma~\ref{lemma:prefixSamplingOne}, we  can finally state the exponential-size model-trace property for $\AAbarBBbarEbar$. In the proof of Theorem \ref{theorem:singleExpTrackModel} below, it is shown  how to derive, from any trace $\rho$ of $\Ku$, an $h$-prefix  bisimilar trace $\rho'$ \emph{induced by} $\rho$ (in the sense that $\rho'$ is obtained by contracting $\rho$, i.e., by concatenating  subtraces of $\rho$ in an ordered way) such that $|\rho'|\leq (|S|\cdot 2^{(2|\SPEC|)^2})^{h+2}$. By Proposition~\ref{prop:fulfillmentPreservingPrefix}, $\rho'$ is indistinguishable from $\rho$ w.r.t.\ the fulfilment of any $\AAbarBBbarEbar$ formula $\varphi$ over the set of atomic formulas in $\SPEC$ such that $\depthb(\varphi)\leq h$.
We preliminarily define the notion of \emph{induced trace} (note that if $\pi$ is induced by $\rho$, then $\fst(\pi)=\fst(\rho)$, $\lst(\pi)=\lst(\rho)$, $|\pi|\leq |\rho|$, and $|\pi| = |\rho|$ iff $\pi = \rho$).

\begin{definition}[Induced trace] \label{definition:inducedTrk}
Let $\rho$ be a trace of $\Ku$ of length $n$. A \emph{trace induced by $\rho$} is a trace $\pi$ of $\Ku$ such that there exists an increasing sequence of $\rho$-positions $i_1<\ldots < i_k$, with $i_1=1$, $i_k=n$, and
$\pi= \rho(i_1)\cdots \rho(i_k)$.
\end{definition}


\begin{theorem}[Exponential-size model-trace property for $\AAbarBBbarEbar$]\label{theorem:singleExpTrackModel}
Let $\rho$ be a trace of $\Ku$ and $h\geq 0$.
Then there exists a trace $\rho'$ induced by $\rho$, whose length is at most $(|S|\cdot 2^{(2|\SPEC|)^2})^{h+2}$, which is $h$-prefix bisimilar to $\rho$. In particular,  for every $\AAbarBBbarEbar$ formula $\psi$ with atomic formulas in $\SPEC$ and such that $\depthb(\psi)\leq h$, it holds that $\Ku,\rho\models \psi$ iff $\Ku,\rho'\models \psi$.
\end{theorem}
\begin{proof} We show that if $|\rho|>(|S|\cdot 2^{(2|\SPEC|)^2})^{h+2}$, then there exists a trace $\rho'$ induced by $\rho$ such that $|\rho'|<|\rho|$
 and $\rho$ and $\rho'$ have the same $h$-sampling word. Hence, by iterating the reasoning and applying
 Proposition~\ref{prop:fulfillmentPreservingPrefix} and Lemma~\ref{lemma:prefixSamplingOne}, the thesis follows.
 
 Assume that $|\rho|>(|S|\cdot 2^{(2|\SPEC|)^2})^{h+2}$.
 Let $\PrefS_{h}: 1=i_1<\ldots <i_N=|\rho|$ be the $h$-prefix sampling of $\rho$. By Property~\ref{property:prefSamBound}, $|\PrefS_{h}|\leq (|S|\cdot 2^{(2|\SPEC|)^2})^{h+1}$.
  Since  the number of distinct summaries (w.r.t.\ $\SPEC$) associated with the prefixes of $\rho$ is at most $|S|\cdot 2^{(2|\SPEC|)^2}$, there must be two consecutive positions $i_j$ and
  $i_{j+1}$ in $\PrefS_h$ such that for some $\ell,\ell'\in [i_j+1,i_{j+1}-1]$ with $\ell<\ell'$, $\Summary(\rho(1,\ell))=\Summary(\rho(1,\ell'))$. It easily follows that
  the sequence $\rho'$ given by $\rho':=\rho(1,\ell)\cdot \rho(\ell'+1,|\rho|)$ is a trace induced by $\rho$ such that $|\rho'|<|\rho|$ and $\rho$ and $\rho'$ have the same
  $h$-sampling word.
\end{proof}

%% file: ModelCheckingAAbarBBbarEbar.tex
\section{$\LINAEXPTIME$-membership of MC for $\AAbarBBbarEbar$}\label{sec:UpperBound}

 In this section, we  exploit the exponential-size model-trace property of  $\AAbarBBbarEbar$ to design a MC algorithm for $\AAbarBBbarEbar$  belonging to the class $\LINAEXPTIME$, namely, the class of problems decidable by  singly exponential-time bounded Alternating Turing Machines (ATMs, for short) with a polynomial-bounded number of alternations. More formally, an ATM $\M$ 
 (we refer to \cite{CKS81} or \cite{preprintGand} for standard syntax and semantics of ATMs)  
 is \emph{singly exponential-time bounded} if there is an integer constant $c\geq 1$ such that for each input $\alpha$, any computation starting on  $\alpha$
 halts after at most $2^{|\alpha|^{c}}$ steps. The ATM $\M$ has a  \emph{polynomial-bounded number of alternations} if there is an integer constant $c\geq 1$ such that, for all inputs $\alpha$ and computations $\pi$ starting from $\alpha$, the number of alternations of existential and universal configurations along $\pi$ is at most $|\alpha|^{c}$.

In the sequel, we assume that 
$\AAbarBBbarEbar$  formulas are in \PNF. For a  formula $\varphi$,
let $\SPEC$ be the set of regular expressions occurring in $\varphi$. 
The size $|\varphi |$ of $\varphi$ is given by the number of non-atomic subformulas of $\varphi$, plus $|\SPEC |$.
As another complexity measure of an $\AAbarBBbarEbar$ formula $\varphi$, we consider 
the standard  \emph{alternation depth}, denoted by $\AltN(\varphi)$, between the existential $\tuple{X}$  and universal
modalities $[X]$ (and vice versa) occurring in the \PNF\ of $\varphi$, for $X\in \{\overline{B},\overline{E}\}$. Note that the definition does not consider the modalities associated with the Allen's relations in $\{A,\overline{A},B\}$. 
Moreover, let $\FMC$ be the set of pairs $(\Ku,\varphi)$ consisting of a Kripke structure $\Ku$ and an $\AAbarBBbarEbar$ formula $\varphi$ such that $\Ku\models\varphi$. The complexity upper bound is as follows.


\begin{theorem}\label{Theorem:UpperBoundAAbrBBarEbar} One can construct a singly exponential-time bounded ATM accepting $\FMC$ whose number of alternations on an input $(\Ku,\varphi)$ is at most $\AltN(\varphi)+2$.
\end{theorem}

\begin{figure}
\begin{bigbox}

\noindent $\textit{check}(\Ku,\varphi)$ \quad [\emph{$\Ku$ is a finite Kripke structure and $\varphi$ is an $\AAbarBBbarEbar$ in \PNF}]

\hrulefill

\noindent \texttt{existentially choose} an $\AAbar$-labeling  $\GLab$ for  $(\Ku,\varphi)$;\\
\noindent \texttt{for each} state $s$ and $\psi\in\GLab(s)$ \texttt{do} \\
\noindent \text{\quad} \texttt{case} $\psi=\hsA \psi'$ (resp., $\psi=\hsAt \psi'$):  \texttt{existentially choose} a certificate $\rho$ with \\
\text{\quad} \phantom{\texttt{case}:}  $\fst(\rho)=s$ (resp., $\lst(\rho)=s$) and \texttt{call} $\textit{checkTrue}_{\tpl{\Ku,\varphi,\GLab}}(\{(\psi',\rho)\})$;\\
 \noindent \text{\quad} \texttt{case} $\psi=[A] \psi'$ (resp., $\psi=[\overline{A}] \psi'$):  \texttt{universally choose} a certificate $\rho$ with \\
 \text{\quad}\phantom{\texttt{case}:}  $\fst(\rho)=s$ (resp., $\lst(\rho)=s$) and \texttt{call} $\textit{checkTrue}_{\tpl{\Ku,\varphi,\GLab}}(\{(\psi',\rho)\})$;\\
 \noindent \texttt{end for}\\
 \noindent \texttt{universally choose} a certificate $\rho$ for $(\Ku,\varphi)$ with $\fst(\rho)=s_0$ ($s_0$ is the initial state of $\Ku$)\\
 \text{\quad\quad\quad}  and  \texttt{call} $\textit{checkTrue}_{\tpl{\Ku,\varphi,\GLab}}(\{(\varphi,\rho)\})$;

\end{bigbox}
\vspace{-0.5cm}
\caption{Procedure $\textit{check}$}
\label{fig-proc-check}
\end{figure}

In the rest of the section, we define a procedure (Figure~\ref{fig-proc-check})---which can be easily translated into an ATM---proving the assertion of Theorem~\ref{Theorem:UpperBoundAAbrBBarEbar}. 
We start with some auxiliary notation. Fix a finite Kripke structure $\Ku$ with set of states $S$ and an $\AAbarBBbarEbar$ formula $\varphi$ in \PNF.
Let $h=\depthb(\varphi)$, and
$\SPEC$ be the set of regular expressions occurring in $\varphi$.

A \emph{certificate} of $(\Ku,\varphi)$ is a trace $\rho$ of $\Ku$ whose length is less than $(|S|\cdot 2^{(2|\SPEC|)^2})^{h+2}$ (the bound for the exponential trace property in Theorem~\ref{theorem:singleExpTrackModel}).
A \emph{$\overline{B}$-witness} (resp., \emph{$\overline{E}$-witness}) of a certificate $\rho$ for $(\Ku,\varphi)$ is a certificate $\rho'$ of  $(\Ku,\varphi)$ such that $\rho'$ is $h$-prefix bisimilar to a trace of the form $\rho\star \rho''$ (resp., $\rho''\star \rho$) for some
\emph{certificate} $\rho''$ of $(\Ku,\varphi)$ with $|\rho''|>1$. By $\SD(\varphi)$ we denote the set consisting of the subformulas $\psi$ of $\varphi$ and  the \emph{duals} $\widetilde{\psi}$.
By the results of Section~\ref{sec:AAbarBBbarEbarTrackProperty}, we deduce the following: 

\newcounter{prop-EbarBbarWitness}
\setcounter{prop-EbarBbarWitness}{\value{proposition}}

\begin{proposition}\label{prop:EbarBbarWitness} Let $\Ku$ be a finite Kripke structure,  $\varphi$ be an $\AAbarBBbarEbar$ formula  in \PNF, and $\rho$ be a certificate for $(\Ku,\varphi)$. The following properties hold:
\begin{compactenum}
  \item for each $\tuple{X}\psi\in\SD(\varphi)$ with $X\in \{\overline{B}, \overline{E}\}$, $\Ku,\rho\models \tuple{X}\psi$ iff there exists an $X$-witness $\rho'$ of $\rho$
  for $(\Ku,\varphi)$ such that $\Ku,\rho'\models \psi$;
  \item for each trace of the form $\rho\star\rho'$ (resp., $\rho'\star\rho$) such that $\rho'$ is a certificate for $(\Ku,\varphi)$, one can construct in time singly exponential in the size of $(\Ku,\varphi)$,
  a certificate $\rho''$ which is $h$-prefix bisimilar to $\rho\star\rho'$ (resp., $\rho'\star\rho$), with $h=\depthb(\varphi)$.
\end{compactenum}
\end{proposition}

 The set $\AAbar(\varphi)$ is the set of formulas in $\SD(\varphi)$ of the form $\tuple{X}\psi'$ or $[X]\psi'$ with $X\in \{A,\overline{A}\}$.
%
%
%
An $\AAbar$-labeling $\GLab$ for $(\Ku,\varphi)$ is a mapping associating to each state $s$ of $\Ku$ a maximally consistent set of subformulas of $\AAbar(\varphi)$. More precisely, for all $s \in S$, $\GLab(s)$ is such that for all $\psi,\widetilde{\psi}\in \AAbar(\varphi)$, $\GLab(s)\cap \{\psi,\widetilde{\psi}\}$ is a singleton.
 We say that $\GLab$ is \emph{valid} if for all states $s \in S$ ad $\psi\in \GLab(s)$, $\Ku,s\models \psi$ (we consider $s$ as a length-1 trace). 
%
Finally, a \emph{well-formed set for $(\Ku,\varphi)$} is a finite set $\WS$ consisting of pairs $(\psi,\rho)$ such that $\psi\in\SD(\varphi)$ and $\rho$ is a certificate of $(\Ku,\varphi)$.   We say that $\WS$ is \emph{universal}
if each formula occurring in $\WS$ is of the form $[X]\psi$ with $X\in \{\overline{B},\overline{E}\}$.  The \emph{dual} $\widetilde{\WS}$ of $\WS$ is the well-formed set  obtained by replacing each pair $(\psi,\rho)\in \WS$ with
   $(\widetilde{\psi},\rho)$.  A well-formed set $\WS$ is \emph{valid} if for each $(\psi,\rho)\in \WS$,
  $\Ku,\rho\models \psi$.

\begin{figure}
\begin{bigbox}
\noindent $\textit{checkTrue}_{\tpl{\Ku,\varphi,\GLab}}(\WS)$  \quad [\emph{$\WS$ is a well-formed set and $\GLab$ is an $\AAbar$-labeling for $(\Ku,\varphi)$}]

\hrulefill

\noindent  \texttt{while} $\WS$ is \emph{not} universal \texttt{do}\\
\noindent \text{\quad} \texttt{deterministically select} $(\psi,\rho)\in \WS$ such that $\psi$ is not of the form  $[\overline{E}]\psi'$ and $[\overline{B}]\psi'$ \\
\noindent \text{\quad} \texttt{update} $\WS\leftarrow \WS\setminus \{(\psi,\rho)\}$;\\
\noindent \text{\quad} \texttt{case} $\psi=r$  with $r\in \RE$:  \texttt{if} $\rho\notin \Lang(r)$   \texttt{then} \emph{reject the input};\\
\noindent \text{\quad} \texttt{case} $\psi=\neg r$ with $r\in \RE$:  \texttt{if} $\rho\in \Lang(r)$  \texttt{then} \emph{reject the input};\\
\noindent \text{\quad} \texttt{case} $\psi=\hsA \psi'$ or $\psi=[A] \psi'$:  \texttt{if} $\psi\notin\GLab(\lst(\rho))$   \texttt{then} \emph{reject the input};\\
\noindent \text{\quad} \texttt{case} $\psi=\hsAt \psi'$ or $\psi=[\overline{A}] \psi'$:  \texttt{if} $\psi\notin\GLab(\fst(\rho))$   \texttt{then} \emph{reject the input};\\
\noindent \text{\quad} \texttt{case} $\psi=\psi_1\vee \psi_2$: \texttt{existentially choose} $i=1,2$, \texttt{update} $\WS\leftarrow \WS\cup \{(\psi_i,\rho)\}$;\\
\noindent \text{\quad} \texttt{case} $\psi=\psi_1\wedge \psi_2$: \texttt{update} $\WS\leftarrow \WS\cup \{(\psi_1,\rho),(\psi_2,\rho)\}$;\\
\noindent \text{\quad} \texttt{case} $\psi=\hsB\psi'$: \texttt{existentially choose} $\rho'\in \Pref(\rho)$, \texttt{update} $\WS\leftarrow \WS\cup \{(\psi',\rho')\}$;\\
\noindent \text{\quad} \texttt{case} $\psi=[B]\psi'$: \texttt{update} $\WS\leftarrow \WS\cup \{(\psi',\rho')\mid \rho'\in\Pref(\rho)\}$;\\
\noindent \text{\quad} \texttt{case} $\psi=\tuple{X}\psi'$ with $X\in \{\overline{E},\overline{B}\}$: \texttt{existentially choose} an $X$-\emph{witness} $\rho'$ of $\rho$  \\
\noindent \text{\quad} \phantom{\texttt{case} $\psi=\tuple{X}\psi'$ with $X\in \{\overline{E},\overline{B}\}$:} for $(\Ku,\varphi)$, \texttt{update}  $\mathcal{\WS}\leftarrow \mathcal{\WS}\cup \{(\psi',\rho')\}$;\\
 \texttt{end while}\\
 \texttt{if} $\mathcal{\WS}=\emptyset$  \texttt{then} \emph{accept}\\
\texttt{else} \texttt{universally choose} $(\psi,\rho)\in \widetilde{\WS}$ and  \texttt{call} $\textit{checkFalse}_{\tpl{\Ku,\varphi,\GLab}}(\{(\psi,\rho)\})$
\end{bigbox}
\vspace{-0.5cm}
\caption{Procedure $\textit{checkTrue}$}
\label{fig-proc-checkTRUE}
\end{figure}

The procedure $\textit{check}$, reported in Figure~\ref{fig-proc-check}, defines the ATM required to prove the assertion of Theorem~\ref{Theorem:UpperBoundAAbrBBarEbar}.
The procedure $\textit{check}$ takes a pair $(\Ku,\varphi)$ as input and: $(1)$ it guesses
  an $\AAbar$-labeling  $\GLab$ for $(\Ku,\varphi)$; $(2)$ it checks that the guessed labeling $\GLab$ is valid; $(3)$ for every certificate $\rho$ starting from the initial state, it checks that $\Ku,\rho\models \varphi$. To perform steps $(2)$--$(3)$, it exploits the
  auxiliary ATM procedure $\textit{checkTrue}$ reported in Figure~\ref{fig-proc-checkTRUE}. The procedure $\textit{checkTrue}$ takes as input a well-formed set  $\WS$ for $(\Ku,\varphi)$ and, assuming that the current $\AAbar$-labeling $\GLab$ is valid,  checks whether $\WS$ is valid. For each pair $(\psi,\rho) \in \WS$ such that $\psi$ is not of the form $[X]\psi'$, with $X\in \{\overline{B},\overline{E}\}$, $\textit{checkTrue}$ directly checks  whether  $\Ku,\rho\models \psi$. 
  In order to allow a deterministic choice of the current element of the iteration, we assume that the set $\WS$ is implemented as an ordered data structure.
  At each iteration of the while loop in $\textit{checkTrue}$, 
  the current pair  $(\psi,\rho)\in \WS$ is processed
   according to the semantics of $\HS$, exploiting the guessed $\AAbar$-labeling  $\GLab$ and Proposition~\ref{prop:EbarBbarWitness}.
   The processing is either deterministic or based on an existential choice,
   and the currently processed pair $(\psi,\rho)$  is either removed from $\WS$, or replaced with pairs $(\psi',\rho')$ such that $\psi'$ is a strict subformula of $\psi$.
   
   At the end of the while loop, the resulting well formed set $\WS$ is either empty or universal. In the former case, the procedure accepts. In the latter case, 
   %
   %
   there is a switch in the current operation mode. For each element $(\psi,\rho)$ in the dual of $\WS$ (note that the root modality  of $\psi$ is either $\hsEt$ or $\hsBt$), the auxiliary ATM procedure $\textit{checkFalse}$ is invoked, which accepts the input $\{(\psi,\rho)\}$ iff $\Ku,\rho \not\models \psi$.
  The procedure  $\textit{checkFalse}$ is the \lq\lq dual\rq\rq{} of $\textit{checkTrue}$: it is simply obtained from $\textit{checkTrue}$ by switching \emph{accept} and \emph{reject},  by switching existential choices and universal choices, and by converting the last call to $\textit{checkFalse}$ into $\textit{checkTrue}$. 
  Thus $\textit{checkFalse}$  accepts an input $\WS$ iff $\WS$ is \emph{not} valid.
    
Recall that the length of a certificate is singly exponential in the size of the input $(\Ku,\varphi)$.
Thus, since the number of alternations of the ATM $\textit{check}$  between existential  and universal choices is evidently the number of switches between the calls to the procedures $\textit{checkTrue}$ and $\textit{checkFalse}$ plus two, by Theorem~\ref{theorem:singleExpTrackModel} and Proposition~\ref{prop:EbarBbarWitness}, we can state the following result that directly implies Theorem~\ref{Theorem:UpperBoundAAbrBBarEbar}. 

\newcounter{prop-correctnessATMcheck}
\setcounter{prop-correctnessATMcheck}{\value{proposition}}

\begin{proposition}\label{prop:correctnessATMcheck}The ATM  $\textit{check}$ is a singly exponential-time bounded ATM accepting $\FMC$  whose number of alternations on an input $(\Ku,\varphi)$ is at most $\AltN(\varphi)+2$.
\end{proposition}

\section{$\LINAEXPTIME$-hardness of MC for $\BEbar$}\label{sec:LowerBound}

In this section, we show that the MC problem for the fragment $\BEbar$ is $\LINAEXPTIME$-hard (implying the $\LINAEXPTIME$-hardness of $\AAbarBBbarEbar$). The result is obtained by a polynomial-time reduction from
a variant of the domino-tiling problem for grids with rows and columns of
 exponential length called \emph{alternating multi-tiling problem}. 

An instance of this problem is a tuple $\Instance=\tpl{n,D,D_0,H,V, M,D_\acc}$, where: $n$ is a positive \emph{even} natural number encoded in unary; $D$ is a  non-empty finite set of \emph{domino types}; $D_0\subseteq D$ is a set of \emph{initial domino types}; $H\subseteq D\times D$ and $V\subseteq D\times D$  are the \emph{horizontal} and \emph{vertical matching relations}, respectively;
 $M\subseteq D\times D$ is the \emph{multi-tiling matching relation};
 $D_\acc\subseteq D$ is a set of \emph{accepting domino types}.
A \emph{tiling of $\Instance$} is a map assigning a domino type to each cell of a $2^{n} \times 2^{n}$ squared grid 
coherently with
the horizontal and vertical matching relations. Formally, a  tiling of $\Instance$  is a mapping   $f:[0,2^{n}-1]\times [0,2^{n}-1] \rightarrow D$ such that:
\begin{compactitem}
  \item for all $i,j\in [0,2^{n}-1]\times [0,2^{n}-1]$ with $j<2^{n}-1$, $(f(i,j),f(i,j+1))\in H$; 
  \item  for all $i,j\in [0,2^{n}-1]\times [0,2^{n}-1]$ with $i<2^{n}-1$, $(f(i,j),f(i+1,j))\in V$. 
\end{compactitem}
The \emph{initial condition} $\Init(f)$ of the tiling $f$ is  the content of the first row of $f$, namely  $\Init(f):= f(0,0)f(0,1)\ldots f(0,2^{n}-1)$.
%
%
A \emph{multi-tiling of $\Instance$} is a tuple $\tpl{f_1,\ldots,f_n}$ of $n$ tilings which are coherent w.r.t.\ the multi-tiling matching relation $M$, namely, such that:
\begin{compactitem}
  \item $(i)$ for all  $i,j\in [0,2^{n}-1]\times [0,2^{n}-1]$ and $\ell\in[1,n-1]$,  $(f_\ell(i,j),f_{\ell+1}(i,j))\in M$ (\emph{multi-cell requirement}), and
 $(ii)$  $f_n(2^{n}-1,j)\in D_\acc$ for some $j\in [0,2^{n}-1]$ (\emph{acceptance}).
\end{compactitem}
The \emph{alternating multi-tiling problem} for an instance $\Instance$ is checking whether
\begin{compactitem}
  \item $\forall w_1\in (D_0)^{2^{n}},\exists w_2 \in (D_0)^{2^{n}},\ldots,\forall w_{n-1}\in (D_0)^{2^{n}}, \exists w_n\in (D_0)^{2^{n}}$ such that there exists a multi-tiling $\tpl{f_1,\ldots,f_n}$ where for all $i\in [1,n]$, $\Init(f_i)=w_i$.
\end{compactitem}


\begin{theorem}[\cite{preprintGand}]\label{theo:ComplexityAlternatingMT}  The alternating multi-tiling problem is $\LINAEXPTIME$-complete.
\end{theorem}
The fact that MC for the fragment $\BEbar$ is $\LINAEXPTIME$-hard is an immediate corollary of the following theorem.

\begin{theorem}\label{theo:MainLowerBoundResult}
One can construct, in time polynomial in the size of $\Instance$, a finite Kripke structure $\Ku_\Instance$ and a $\BEbar$ formula
$\varphi_\Instance$ over the set of propositions $\Prop = D \cup (\{r,c\}\times \{0,1\}) \cup \{\bot,\End\}$
such that  $\Ku_\Instance\models \varphi_\Instance$ iff $\Instance$ is a \emph{positive} instance of the alternating multi-tiling problem.
\end{theorem}


The rest of this section is devoted to the construction of the Kripke structure $\Ku_\Instance$ and the $\BEbar$ formula
$\varphi_\Instance$, proving Theorem~\ref{theo:MainLowerBoundResult}. Let $\Prop$ be  as in the statement
of Theorem~\ref{theo:MainLowerBoundResult}. The Kripke structure $\Ku_\Instance$ is given by
$\Ku_\Instance=(\Prop,S, \Trans,\mu,s_0)$, where $S=\Prop$, $s_0=\End$, $\mu$ is the identity mapping (we identify
a singleton set $\{p\}$ with $p$), and $\Trans = \{(s,s') \mid s\in \Prop\setminus \{\End\}, s'\in\Prop\}$. Note that the initial state $\End$ has no successor,
and that a trace of $\Ku_\Instance$ can be identified with its induced labeling sequence.

The construction of the $\BEbar$ formula
$\varphi_\Instance$ is based on a suitable encoding of multi-tilings which is described in the following. The symbols
$\{r\}\times   \{0,1\}$ and  $\{c\} \times \{0,1\}$ in $\Prop$
are used to encode the values of two $n$-bits counters numbering the $2^{n}$ rows and  columns, respectively, of a  tiling.
For a multi-tiling $F=\tpl{f_1,\ldots,f_n}$ and for all $i,j\in [0,2^{n}-1]$, the $(i,j)$-{th} \emph{multi-cell} $\tpl{f_1(i,j),\ldots,f_n(i,j)}$ of $F$ is encoded by
the word $C$ of length $3n$ over $\Prop$, called \emph{multi-cell code}, given by
$d_1 \cdots d_n (r, b_1)\cdots(r, b_n)(c,b'_1)\cdots(c, b'_n)$
where $b_1 \cdots b_n$ and $b'_1\cdots b'_n$ are the binary encodings of the row number $i$ and column number $j$, respectively, and for all $\ell\in[1,n]$, $d_\ell= f_\ell(i,j)$ (i.e., the content of the $(i,j)$-th cell of component $f_\ell$).
The \emph{content} of  $C$ is $d_1\cdots d_n$. 
Since $F$ is a multi-tiling,  the following well-formedness requirement must  be satisfied by the encoding $C$: for all $\ell\in [1,n-1]$, $(d_\ell,d_{\ell+1})\in M$. We call such words \emph{well-formed multi-cell codes}.

\begin{definition}[Multi-tiling codes]\label{Def:multiTilingCodes} A \emph{multi-tiling code} is a finite word $w$ over $\Prop$  obtained by concatenating well-formed multi-cell codes in such a way that the following conditions hold:
\begin{compactitem}
  \item for all $i,j\in [0,2^{n}-1]$, there is a multi-cell code in $w$  with row number $i$ and column number $j$ \emph{(completeness requirement)};
  \item for all multi-cell codes $C$ and $C'$ occurring in $w$, if $C$ and $C'$ have the same row number and column number, then $C$ and $C'$ have the same content \emph{(uniqueness requirement)};
   \item for all multi-cell codes $C$ and $C'$ in $w$ having the same row-number
   (resp., column number), column numbers  (resp., row numbers) $j$ and $j+1$, respectively, and contents $d_1\cdots d_n$ and $d'_1\cdots d'_n$, respectively, it holds that $(d_\ell,d'_\ell)\in H$ (resp. $(d_\ell,d'_\ell)\in V$) for all $\ell\in [1,n]$ \emph{(row-adjacency requirement)} (resp., \emph{(column-adjacency requirement)});
  \item there is a multi-cell code in $w$ with row-number $2^{n}-1$   whose content is in $D^{n-1}\cdot d_\acc$ for some $d_\acc\in D_\acc$ \emph{(acceptance requirement)}.
\end{compactitem}
\end{definition}

Finally, we have to encode the initial conditions of the components of a multi-tiling.
An \emph{initial cell code} encodes a cell of the first row of a tiling
%
and is a word $w$ of length $n+1$ of the form $w =d (c,b_1) \cdots (c,b_n)$, where $d\in D_0$ and $b_1,\ldots,b_n\in \{0,1\}$. We say that $d$ is the \emph{content} of $w$ and the integer in $[0,2^{n}-1]$ encoded by $b_1\cdots b_n$ is the \emph{column number} of $w$.

\begin{definition}[Multi-initialization codes]\label{Def:InitializationCodes} An \emph{initialization code} is a finite word $w$ over $\Prop$ which is the concatenation of initial cell codes such that:
\begin{compactitem}
  \item for all $i \in [0,2^{n}-1]$, there is an initial cell code in $w$  with column number $i$. 
  \item for all initial cell codes $C$ and $C'$ occurring in $w$, if $C$ and $C'$ have the same  column number, then $C$ and $C'$ have the same content.
\end{compactitem}
A \emph{multi-initialization code} is a finite word over $\Prop$ of the form $\bot\cdot w_n\cdots \bot \cdot w_1\cdot \End$
such that for all $\ell\in [1,n]$, $w_\ell$ is an initialization code.
\end{definition}

\begin{definition}[Initialized multi-tiling codes]\label{Def:IMTCodes} An \emph{initialized multi-tiling code} is a finite word over $\Prop$ of the form
$\bot \cdot w \cdot \bot\cdot w_n\cdots \bot \cdot w_1\cdot \End$
such that $w$ is a multi-tiling code, $\bot\cdot w_n\cdots \bot \cdot w_1\cdot \End$ is a multi-initialization code, and the following requirement holds:
\begin{compactitem}
   \item for each multi-cell code in $w$ having row number $0$, column number $i$, and content $d_1\cdots d_n$ and for all $\ell\in [1,n]$, there is
  an initial cell code in $w_\ell$ having column number $i$ and content $d_\ell$  \emph{(initialization coherence requirement)}.
\end{compactitem}
\end{definition}

\newcounter{prop-FormulasForMultiTilingCodes}
\setcounter{prop-FormulasForMultiTilingCodes}{\value{proposition}}

We sketch now the idea for the construction of the $\B\Ebar$ formula $\varphi_\Instance$ ensuring that $\Ku_\Instance\models \varphi_\Instance$ iff $\Instance$ is a positive instance of the alternating multi-tiling problem. We preliminarily  observe that  since the initial state of $\Ku_\Instance$ has no successors, the only initial trace of $\Ku_\Instance$ is the trace $end$ of length 1. To guess a trace corresponding to an initialized multi-tiling code, $\Ku_\Instance$ is unraveled backward starting from $end$, exploiting the modality $\Ebar$. The structure of the formula $\varphi_\Instance$ is 
\[
\varphi_\Instance:= [\overline{E}](\varphi_1 \rightarrow \hsEt(\varphi_2\wedge(\ldots ([\overline{E}](\varphi_{n-1} \rightarrow \hsEt(\varphi_n\wedge \hsEt \varphi_\IMT)))\ldots))).
\]
The formula $\varphi_\Instance$ features $n+1$ unravelling steps starting from the initial trace $end$. The first $n$ steps are used to guess a sequence of $n$ initialization codes. Intuitively, each formula $\varphi_i$ is used to constrain 
the $i$-th unravelling to be an initialization code, in such a way that at depth $n$ in the formula a multi-initialization code is under evaluation. 
The last unravelling step (the innermost in the formula) is used to guess the multi-tiling code. Intuitively, the innermost formula $\varphi_\IMT$ is evaluated over a trace corresponding to an initialized multi-tiling code, and checks its structure: 
multi-cell codes are \lq\lq captured\rq\rq{} by regular expressions (encoding in particular their row and column numbers and contents);
moreover the completeness, uniqueness, row- and column-adjacency requirements of Definition~\ref{Def:multiTilingCodes} are enforced by the joint use of $[\overline{E}]$ and regular expressions:
intuitively, by means of $[\overline{E}]$, one or two multi-cell codes are generated \lq\lq separately\rq\rq; then, if they appear in the considered multi-tiling code, the aforementioned constraints are verified by means of auxiliary formulas, consisting of suitable regular expressions.
The initialization coherence requirement of Definition~\ref{Def:IMTCodes} is guaranteed in an analogous way, by comparing initial cell codes and multi-cell codes.
Note that the first $n-1$ occurrences of alternations between universal and existential modalities  $[\overline{E}]$ and $\hsEt$ correspond to the alternations of universal and existential quantifications in the definition of alternating multi-tiling problem. 

Proposition~\ref{prop:FormulasForMultiTilingCodes} states the correctness of the construction of $\varphi_\Instance$ (for the definitions of $\varphi_1 ,\ldots, \varphi_n$, and $\varphi_\IMT$, see \cite{preprintGand}).

\begin{proposition}\label{prop:FormulasForMultiTilingCodes}  One can build, in  time polynomial in the size of $\Instance$,
$n+1$ $\BEbar$ formulas $\varphi_\IMT,\varphi_1,\ldots,\varphi_n$ such that $\AltN(\varphi_\IMT)=\AltN(\varphi_1) =\ldots =\AltN(\varphi_n)= 0$, and fulfilling the following conditions.
\begin{compactitem}
  \item For all finite words $\rho$ over $\Prop$ of the form $\rho =\rho' \cdot \bot \cdot w_n  \cdots  \bot \cdot w_1\cdot \End$ such that $\rho'\neq \varepsilon$ and
  $\bot \cdot w_n  \cdots  \bot \cdot w_1\cdot \End$ is a multi-initialization code, $\Ku_\Instance,\rho\models\varphi_\IMT$ if and only if
  $\rho$ is an initialized multi-tiling code.
  \item For all $\ell\in [1,n]$ and words $\rho$ of the form $\rho=\rho'\cdot \bot \cdot w_{\ell-1} \cdots \bot \cdot w_1\cdot \End$ such that
 $\rho'\neq \varepsilon$ and $w_j\in (\Prop\setminus \{\bot\})^{*}$ for all $j\in [1,\ell-1]$, $\Ku_\Instance,\rho\models\varphi_\ell$ if and only if
$\rho'$ is of the form  $\rho'=\bot\cdot w_\ell$, where $w_\ell$ is an initialization code.
\end{compactitem}
\end{proposition}

Since the initial state of $\Ku_\Instance$ has no successors and corresponds to the atomic proposition $\End$, by
Proposition~\ref{prop:FormulasForMultiTilingCodes} and Definitions~\ref{Def:multiTilingCodes}--\ref{Def:IMTCodes}, we obtain that
$\Ku_\Instance\models \varphi_\Instance$ iff $\Instance$ is a positive instance of the alternating multi-tiling problem.
This concludes the proof of Theorem~\ref{theo:MainLowerBoundResult}. 

%% file: concl.tex
\section{Conclusions and future work}
In this paper, we have investigated the MC problem for two maximal fragments of $\HS$, $\AAbarBBbar\Ebar$ and $\AAbar\E\Bbar\Ebar$, endowed with 
interval labeling based on regular expressions, and we have proved that such a problem is 
$\LINAEXPTIME$-complete. The paper also settles, in the more general setting of the regular expression-based semantics,
the open complexity question for the same fragments under the homogeneity assumption.
Future work will focus on the problem of determining the exact complexity of MC for full $\HS$, both under homogeneity and in the regular expression-based semantics.
In addition, we will study the MC problem for $\HS$ over \emph{visibly pushdown systems} (VPS), in order to deal with recursive programs and infinite state systems.
Finally, we are thinking of inherently \emph{interval-based models of systems}. Kripke structures, being based on states, are naturally oriented to the description of point-based properties of systems, and of how they evolve state-by-state.
We want to come up with suitable (and practical) description paradigms for systems, which allow us to directly model them on the basis of their interval behavior/properties. Only after devising these models (something that seems to be extremely challenging), a really general interval-based MC will be possible.